%% file: bin-subsidy.tex
\documentclass[11pt, english]{article}

\input{includes}

\usepackage{amsmath}
\usepackage{amssymb,amsthm,mathtools}

\usepackage{amsmath,amsfonts} 
\usepackage{caption} 
\usepackage{xcolor}
\usepackage{hyperref}
\hypersetup{
	colorlinks   = true,
	linkcolor    = DarkBlue, 
	urlcolor     = DarkBlue, 
	citecolor    = DarkGreen 
}

\usepackage{comment}
\usepackage{xfrac}

\usepackage{verbatim}

\usepackage[font=small,labelfont=bf]{caption}
\usepackage{xspace}
\usepackage{etoolbox}
\usepackage{booktabs}

\newcommand{\A}{\mathcal{A}}

\newcommand{\instance}[1]{\langle {#1} \rangle}
\newcommand{\val}{\{v_i\}_{i \in [n]}}

\newcommand{\M}{{\rm M}}
\newcommand{\EFone}{\textsc{EF1}}
\newcommand{\EF}{\textsc{EF}}

\allowdisplaybreaks

\title{\bfseries Achieving Envy-Freeness with Limited Subsidies under Dichotomous Valuations}

\author{Siddharth Barman\thanks{Indian Institute of Science. {\tt barman@iisc.ac.in}} \qquad Anand Krishna\thanks{Indian Institute of Science. {\tt anandkrishna@iisc.ac.in}} \qquad Y.~Narahari\thanks{Indian Institute of Science. {\tt narahari@iisc.ac.in}} \qquad Soumyarup Sadhukhan\thanks{Indian Institute of Technology Kanpur. {\tt soumyarup.sadhukhan@gmail.com}}}

\date{}

\begin{document}
\maketitle

\input{abstract}

\input{introduction}

\input{preliminaries}

\input{results}

\input{algorithm}

\input{conclusion}

\bibliographystyle{alpha} 
\bibliography{references.bib}

\input{appendix}

\end{document}

%% file: includes.tex
\usepackage[margin=0.8in]{geometry}

\usepackage{varioref}
\usepackage[usenames,svgnames,xcdraw,table]{xcolor}
\definecolor{DarkBlue}{rgb}{0.1,0.1,0.5}
\definecolor{DarkGreen}{rgb}{0.1,0.5,0.1}

\usepackage{nicefrac}

\usepackage{hyperref}
\hypersetup{
	colorlinks   = true,
	linkcolor    = DarkBlue, 
	urlcolor     = DarkBlue, 
	citecolor    = DarkGreen 
}

\usepackage{appendix}

\usepackage{mathtools}

\usepackage{algorithmic}
\usepackage{algorithm}
\usepackage{array}

\usepackage[font={small,it}]{caption}

\usepackage{graphicx,epsfig,amsmath,latexsym,amssymb,verbatim}

\usepackage[square,sort,semicolon,numbers]{natbib}

\newcommand{\extra}[1]{}

\usepackage{amsmath,amssymb,amsfonts}
\usepackage{amsthm}
\usepackage{thmtools,thm-restate}
\usepackage{enumitem}

\newtheorem{theorem}{Theorem}

\newtheorem{definition}{Definition}
\newtheorem{lemma}[theorem]{Lemma}

\newtheorem{claim}{Claim}

\def\squareforqed{\hbox{\rlap{$\sqcap$}$\sqcup$}}
\def\qed{\ifmmode\squareforqed\else{\unskip\nobreak\hfil
	\penalty50\hskip1em\null\nobreak\hfil\squareforqed
	\parfillskip=0pt\finalhyphendemerits=0\endgraf}\fi}
\def\endenv{\ifmmode\;\else{\unskip\nobreak\hfil
	\penalty50\hskip1em\null\nobreak\hfil\;
	\parfillskip=0pt\finalhyphendemerits=0\endgraf}\fi}

\renewenvironment{proof}{\noindent \textbf{{Proof~} }}{\qed\medskip}
\newenvironment{proof+}[1]{\noindent \textbf{{Proof #1~} }}{\qed\medskip}

\mathchardef\ordinarycolon\mathcode`\:
\mathcode`\:=\string"8000
\def\vcentcolon{\mathrel{\mathop\ordinarycolon}}
\begingroup \catcode`\:=\active
\lowercase{\endgroup
\let :\vcentcolon
}


%% file: abstract.tex
\begin{abstract}
We study the problem of allocating indivisible goods among agents in a fair manner. While envy-free allocations of indivisible goods are not guaranteed to exist, envy-freeness can be achieved by additionally providing some subsidy to the agents. These subsidies can be alternatively viewed as a divisible good (money) that is fractionally assigned among the agents to realize an envy-free outcome. In this setup, we bound the subsidy required to attain envy-freeness among agents with dichotomous valuations, i.e., among agents whose marginal value for any good is either zero or one.  

We prove that, under dichotomous valuations, there exists an allocation that achieves envy-freeness with a per-agent subsidy of either $0$ or $1$. Furthermore, such an envy-free solution can be computed efficiently in the standard value-oracle model. Notably, our results hold for general dichotomous valuations and, in particular, do not require the (dichotomous) valuations to be additive, submodular, or even subadditive. Also, our subsidy bounds are tight and provide a linear (in the number of agents) factor improvement over the bounds known for general monotone valuations. 
\end{abstract}

%% file: introduction.tex
\section{Introduction}
Discrete fair division is an extremely active field of work at the interface of mathematical economics and computer science~\cite{handbook2016,endriss2017trends}. This field addresses the problem of finding fair allocations of indivisible goods (i.e., resources that cannot be fractionally assigned) among agents with individual preferences. Solution concepts and algorithms developed here address several real-world settings, such as assignment of housing units \cite{deng2013story,benabbou2020finding}, course allocation \cite{Budish2017CourseMA}, and inheritance division; the widely-used platform spliddit.org~\cite{Spliddit} provides fair-division methods for a range of other allocation problems.  

The quintessential and classic \cite{foley1966resource,varian1974equity} notion of fairness in mathematical economics is that of envy-freeness, which requires that each agent prefers the bundle assigned to her over that of any other agent. However, in the case of indivisible goods, an envy-free allocation is not guaranteed to exist; consider the (ad nauseam) example of two agents and one good. 

Motivated, in part, by such considerations, a notable thread of research in discrete fair division aims to achieve the ideal of envy-freeness with the use of subsidies~\cite{halpern2019fair,brustle2019dollar,goko2021fair,caragiannis2020computing}. The objective  is to provide each agent $i$---in addition to a bundle---an appropriate subsidy $p_i$ such that envy-freeness is achieved overall. That is, the aim is to identify an allocation of the goods and subsidies such that each agent $i$'s value for her bundle plus her subsidy, $p_i$, is at least as large as $i$'s value for any other agent $j$'s bundle plus the subsidy $p_j$.

The subsidy model can be alternatively viewed as a setting wherein we have a divisible good (money), along with the indivisible ones. The divisible good can be fractionally assigned among the agents towards achieving fair (envy-free) outcomes. A natural question in this context is to understand  how much subsidy is required to eliminate all envy. This question was in fact addressed in the classic work of Maskin \cite{Maskin1987O} for the case of unit-demand agents. Considering a scaling, wherein each agent has value at most $1$ for any good, Maskin showed that a total subsidy of $(n-1)$ suffices;\footnote{Here, a valuation scaling is unavoidable, since the obtained guarantee is an absolute bound (on the total subsidy).} throughout, $n$ denotes the number of agents participating in the fair division exercise. This unit-demand setting is referred to as the rent division problem and has been studied over the past few decades; see, e.g.,  \cite{aziz2020developments} and references therein. Going beyond rent division, the above-mentioned results address settings in which agents' preferences span the subsets of the goods. In particular, the work of Halpern and Shah~\cite{halpern2019fair} shows that, if the $n$ agents have binary additive valuations, then envy-freeness can be achieved with a total subsidy of $(n-1)$. A similar result was obtained for binary submodular valuations by Goko et al.~\cite{goko2021fair}. 

The current work contributes to this line of work with a focus on valuations that have binary marginals. Specifically, we address fair division instances in which, for each agent $i$, the marginal value of any good $g$ relative to any subset $S$ is either zero or one, $v_i(S \cup \{g\}) - v_i(S) \in \{0,1\}$; here, set function $v_i$ denotes the valuation of agent $i$. Such set functions are referred to as dichotomous valuations and have received significant attention in various fair division settings, e.g.,~\cite{Bogomolnaia05,bouveret2008efficiency,Freitas2013CombinatorialAU,Kurokawa2018LeximinAI,BEF2020}.  Indeed, dichotomous valuations---and restricted subclasses thereof---model preferences in several application domains; see, e.g., \cite{Roth05pairwisekidney,deng2013story,benabbou2020finding}. 

The following stylized example illustrates the applicability of dichotomous valuations and the use of subsidies: consider a scheduling setting in which time-slots (indivisible goods) have be to assigned among employees (agents) whose nonzero marginal values for the slots are the same; equivalently, the agents have binary marginals for the slots. An employee's preference on whether a time-slot can be utilized (i.e., the slot has a nonzero marginal value) could depend on multiple factors, such as contiguity of slots, shift rotations, and time of day. Still, given that the current work holds for general dichotomous valuations,\footnote{In particular, we do not require the dichotomous valuations to be additive, submodular, or even subadditive.} finding a fair allocation in such a scheduling setting falls under the purview of the current work. Also, in this context, subsidies can be viewed as bonuses paid to the employees. \\
 
\noindent
\textbf{Our Results.} We prove that, under dichotomous valuations, there exists an allocation that achieves envy-freeness with a per-agent subsidy of either $0$ or $1$  (Theorem \ref{thm:main}). Furthermore, such an envy-free solution can be computed efficiently in the standard value-oracle model. Our per-agent subsidy bound implies that, for dichotomous valuations, a total subsidy of at most $(n-1)$ suffices to realize envy-freeness. Indeed, this guarantee is tight: consider an instance with $n$ agents and a single unit-valued good. In this instance, to achieve envy-freeness under any allocation, $(n-1)$ agents would each require a subsidy of $1$. \\

\noindent
\textbf{Additional Related Work.} As mentioned previously, Halpern and Shah~\cite{halpern2019fair} showed that if the agents' valuations are binary additive, then envy-freeness can be achieved with a cumulative subsidy of $(n-1)$. The work of Goko et al.~\cite{goko2021fair} provides an analogous result for binary submodular valuations and further develops a strategy-proof mechanism for this setting. Note that binary additive and binary submodular functions constitute subclasses of dichotomous valuations. 

Building upon \cite{halpern2019fair}, Brustle et al.~\cite{brustle2019dollar} proved that a subsidy of at most $1$ per agent suffices for additive valuations; this result is established under a standard scaling, wherein the value of any good across all agents is at most $1$. Brustle et al.~\cite{brustle2019dollar} also obtained an $O(n^2)$ upper bound on the total required subsidy for general monotone valuations. The current work achieves a factor $n$ improvement over this upper bound for dichotomous valuations in particular.

Caragiannis and Ioannidis~\cite{caragiannis2020computing} study the problem of finding allocations that realize envy-freeness with as little subsidy as possible. Since this  problem is {\rm NP}-hard, they focus on designing approximation algorithms and, in particular, develop a fully polynomial-time approximation scheme for instances with a constant number of agents. Narayan et al.~\cite{Narayan2021TwoBW} show that one can achieve envy-freeness with transfers, while incurring a constant-factor loss in Nash social welfare. 

It is also relevant to note that the framework of subsidies complements discrete fair division results which strive towards existential guarantees by relaxing the envy-freeness criterion. Here, a well-studied solution concept (relaxation) is envy-freeness up to one good ($\EFone$) \cite{Budish2011TheCA}: an allocation is said to be $\EFone$ iff each agent values her bundle at least as much as any other agent's bundle, up to the removal of some good from the other agent's bundle. $\EFone$ allocations are known to exist under monotone valuations \cite{LMMS042}.

%% file: preliminaries.tex
\section{Notation and Preliminaries}
We study the problem of allocating $m$ indivisible goods among $n$ agents, with subsidies, in a fair manner. The cardinal preference of each agent $i\in [n]$, over the subsets of goods, is specified via valuation $v_i : 2^{[m]}\to \mathbb{R}_+$. Here, $v_i(S)\in \mathbb{R}_+$ denotes the valuation that agent $i$ has for a subset of goods $S\subseteq [m]$. We will throughout represent a discrete fair division instance via the tuple $ \langle [n],[m], \{v_i\}_{i=1}^n \rangle$. {Our algorithms work in the standard value-oracle model. That is, for each valuation $v_i$, we only require an oracle that---when queried with a subset $S \subseteq [m]$---provides the value $v_i(S)$.}  

This work focuses on valuations that have binary marginals, i.e., for each agent $i \in [n]$, the marginal value of including any good $g \in [m]$ in any subset $S \subseteq [m]$ is either zero or one: $v_i(S\cup \{g\})-v_i(S) \in \{0,1\}$. We will refer to such set functions $v_i$s as dichotomous valuations. Note that a dichotomous valuation $v_i$ is monotone: $v_i(A)\leq v_i(B)$ for all subsets $A\subseteq B\subseteq [m]$. Also, we will assume that the agents' valuations satisfy $v_i(\emptyset) = 0$.

 An allocation $\mathcal{A} = (A_1,A_2,\ldots,A_n)$ is an ordered collection of $n$ pairwise disjoint subsets, $A_1,\ldots, A_n \subseteq [m]$, wherein the subset of goods $A_i$ is assigned to agent $i \in [n]$. We refer to each such subset $A_i$ as a bundle. Note that an allocation can be partial in the sense that $\cup_{i=1}^n A_i \neq [m]$. For disambiguation, we will use the term complete allocation to denote allocations wherein all the goods have been assigned and, otherwise, use the term partial allocation. The social (utilitarian) welfare of an allocation $\A = (A_1, \ldots, A_n)$ is the sum of the values that it generates among the agents, $\sum_{i=1}^n v_i(A_i)$. 
 
The quintessential notion of fairness is that of envy-freeness \cite{foley1966resource}. This notion is defined next for allocations of indivisible goods. 
\begin{definition}
An allocation $\mathcal{A}=(A_1, \ldots, A_n)$ is said to be envy-free ($\EF$) iff $v_i(A_i)\geq v_i(A_j)$ for all agents $i,j\in [n]$.
\end{definition}

Since, in the context of indivisible goods, envy-free allocations are not guaranteed to exist, a realizable desideratum is to achieve envy-freeness with the use of subsidies. The objective here is to provide each agent $i \in [n]$---in addition to a bundle $A_i$---an appropriate subsidy $p_i \geq 0$ such that envy-freeness is achieved overall. We will write vector $p=(p_1,p_2,\ldots, p_n) \in \mathbb{R}_+^n$ to denote the subsidies assigned to the agents. Our overarching aim is to find a complete allocation $\mathcal{A}=(A_1, \ldots, A_n)$ and a (bounded) subsidy vector $p=(p_1,\ldots, p_n)$ that ensure envy-freeness for each agent $i$, i.e., agent $i$'s value for her bundle ($A_i$) plus her subsidy, $p_i$, is at least as large as $i$'s value for any other agent $j$'s bundle ($A_j$) plus the subsidy $p_j$. Formally, 
\begin{definition}[Envy-Free Solution]
An allocation $\mathcal{A}=(A_1, \ldots, A_n)$ and a subsidy vector $p=(p_1, \ldots, p_n)$ are said to constitute an envy-free solution, $(\mathcal{A},p)$, iff $v_i(A_i)+p_i \geq v_i(A_j)+p_j$ for all agents $i,j \in [n]$. 
\end{definition}
The term \emph{envy-freeable} will be used to denote any allocation $\mathcal{A}$ for which there exists a subsidy vector $p$ such that $(\mathcal{A},p)$ is an envy-free solution. 
Also, we will throughout write $\M(p)$ to denote the set of agents that receive maximum subsidy under $p = (p_1, \ldots, p_n)$, i.e., $\M(p) \coloneqq \left\{ i\in [n] \mid  p_i \geq p_j  \text{ for all } j\in [n] \right\}$.

\subsection{Characterizing Envy-Freeable Allocations}
The work of Halpern and Shah \cite{halpern2019fair} provides a useful characterization of envy-freeable allocations.\footnote{An analogous result was  developed by Aragones~\cite{aragones1995derivation} in the context of fair rent division.} Using the characterization, one can determine (in polynomial time) whether or not a given allocation $\mathcal{B}=(B_1, \ldots, B_n)$ is envy-freeable, i.e., whether $\mathcal{B}$ can be coupled with a subsidy vector $q=(q_1, \ldots, q_n)$ to realize an envy-free solution $(\mathcal{B}, q)$. Furthermore, for any envy-freeable allocation $\mathcal{B}$, the characterization leads to a polynomial-time algorithm to compute a corresponding subsidy vector $q$. 

The characterization is obtained by considering a directed graph that captures the envy among the agents. Specifically, for an allocation $\A=(A_1,\ldots, A_n)$ , the \emph{envy-graph} $G_\A$ is a complete, weighted, directed graph wherein the vertices correspond to the agents and each directed edge $(i,j)$ has weight $w_\A(i,j) \coloneqq v_i(A_j)-v_i(A_i)$. Note that the weight of the edge $(i,j)$ is the envy that agent $i$ has towards $j$, and the edge weights can be negative. The weight of any (directed) path $P$ in the graph $G_\A$ will be denoted by $w_\A(P)$; in particular, $w_\A(P)$ is the sum of the weights of the edges along the path, $w_{\A}(P) = \sum_{(i,j)\in P}\ w_\A(i,j)$. 

To achieve envy-freeness in certain settings, we will have to reassign the bundles among the agents. In particular, for an allocation $\mathcal{A} = (A_1, \ldots, A_n)$ and permutation $\sigma\in \mathbb{S}_n$ (over the $n$ agents), write $\mathcal{A}_\sigma \coloneqq (A_{\sigma(1)},A_{\sigma(2)},\ldots,A_{\sigma(n)})$ to denote the allocation wherein each agent $i \in [n]$ receives the bundle $A_{\sigma(i)}$. 

The theorem below states the characterization of envy-freeable allocations. 
\begin{theorem}\cite{halpern2019fair}
\label{thm:hs_ef}
For any allocation $\mathcal{A}=(A_1, \ldots, A_n)$, the following statements are equivalent:
\begin{enumerate}
	\item[(i)]$\mathcal{A}$ is envy-freeable.
	\item[(ii)]$\mathcal{A}$ maximizes the social welfare across all reassignments of its bundles among the agents. 
That is, for every permutation $\sigma$ over $[n]$, we have $ \sum_{i=1}^n v_i(A_i) \geq \sum_{i=1}^n v_i(A_{\sigma(i)}) $.
	\item[(iii)] The envy-graph $G_\mathcal{A}$ has no positive-weight cycles.
\end{enumerate} 
\end{theorem}

Condition $(ii)$ of Theorem~\ref{thm:hs_ef} implies that, starting with an arbitrary allocation $\mathcal{B}=(B_1, \ldots, B_n)$, we can always find an envy-freeable allocation by reassigning the bundles $B_i$s among the agents. In particular, if $\sigma$ is any maximum-weight matching between the agents and the bundles, then $B_\sigma$ is an envy-freeable allocation. 

Complementing the characterization, Theorem \ref{thm:hs_sub} (below) identifies an efficient algorithm to compute a subsidy vector for any given envy-freeable allocation $\A$. Here, $\ell_\A(i)$ is used to denote the maximum weight of any path (in $G_\mathcal{A}$) which starts at $i \in [n]$. Since the envy-graph $G_\A$ (for an envy-freeable allocation $\A$) does not contain any positive-weight cycle (condition $(iii)$ in Theorem \ref{thm:hs_ef}), the weight $\ell_\A(i)$ is well-defined and can be computed efficiently. In particular, we can apply the Floyd-Warshall algorithm, after negating the edge weights. The  theorem asserts that providing a subsidy of $\ell_\A(i)$, to each agent $i$, suffices to achieve envy-freeness. 

\begin{theorem}\label{thm:hs_sub}\cite{halpern2019fair}
For any envy-freeable allocation $\A$, a subsidy of $p^*_i \coloneqq \ell_\A(i)$, for all agents  $i\in [n]$, realizes an envy-free solution $(\A,p^*)$; here subsidy vector $p^*=(p^*_1, \ldots, p^*_n)$. Furthermore, for any other vector $p=(p_1,\ldots, p_n) \in \mathbb{R}_+^n$, such that $(\A, p)$ is envy-free, we have $p^*_i\leq p_i$ for all $i \in [n]$. 
\end{theorem}

For any subsidy vector $p=(p_1, \ldots, p_n)$ and permutation $\sigma$ over $[n]$, let $p_\sigma=(p_{\sigma(1)},p_{\sigma(2)},\ldots,p_{\sigma(n)})$ denote the subsidies obtained by reassigning according to $\sigma$, i.e., under $p_\sigma$, agent $i$ receives a subsidy $p_{\sigma(i)}$. 

The following lemma shows that if reassigning bundles, according to some permutation $\sigma$, maintains envy-freeability, then, in fact, reassigning the subsidies---also according to $\sigma$---preserves envy-freeness. A proof of this lemma appears in Appendix \ref{appendix:lemma-shuffle}.\footnote{We prove Lemma~\ref{lem:reshuffle} for general valuations. Hence, it holds in particular for dichotomous valuations.}
\begin{restatable}{lemma}{LemmaShuffle} 
\label{lem:reshuffle}
Let $ (\mathcal{A},p) $ be an envy-free solution and let $\sigma$ be a permutation such that allocation $ \mathcal{B} \coloneqq \mathcal{A}_\sigma $ is envy-freeable. Then, $(\mathcal{B},p_\sigma)$ is an envy-free solution as well. 
\end{restatable}

%% file: results.tex
\section{Our Results}
The next theorem is the main result of the current work and it shows that in any discrete fair division instance with dichotomous valuations there exists an envy-free solution with a subsidy of at most $1$ per agent. 
\begin{restatable}{theorem}{mainthm}
\label{thm:main}
For any discrete fair division instance $\langle [n],[m],\{v_i\}_{i=1}^n\rangle$ with dichotomous valuations, there exists an envy-free solution $(\mathcal{A},p)$ such that $p\in \{0,1\}^n$. Furthermore, given value oracle access to the $v_i$s, such an envy-free solution, $(\mathcal{A}, p)$, can be computed in polynomial time.
\end{restatable}

This theorem implies that, to achieve envy-freeness under dichotomous valuations, the total required subsidy is at most $(n-1)$, i.e., $\sum_{i=1}^n p_i \leq n-1$. Note that the case in which $p_i = 1$, for all agents $i$, can be addressed by giving each agent a subsidy of $0$.

%% file: algorithm.tex
\section{Finding Envy-Free Solutions for Dichotomous Valuations}
\label{section:main-algorithm}
For dichotomous valuations, our algorithm $\textsc{Alg}$ (Algorithm~\ref{Alg:BinSubsidy}) finds an envy-free solution $(\A,p)$ such that the subsidy vector $p\in \{0,1\}^n$. In particular, with iteration counter $t$, the algorithm inductively maintains an envy-free solution $(\A^t, p^t)$ wherein $\A^t=(A^t_1, \ldots A^t_n)$ is a (partial) envy-freeable allocation and the subsidy vector $p^t \in \{0,1\}^n$.  In every iteration $t$, the algorithm selects an unallocated good $g \in [m] \setminus \left( \cup_{i=1}^n A^t_i \right)$, includes it in one of the current bundles, $A^t_1, \ldots, A^t_n$, and possibly reassigns the bundles among the agents such that the updated allocation, $\A^{t+1}$, is also envy-freeable with subsidies in $\{0, 1\}$. 

We note that a relatively simple case occurs if---in any iteration $t$ and for an unallocated good $g$---there exists an agent $k \in [n]$ such that ({\rm a}) the marginal value of good $g$ for agent $k$ is equal to $1$ (i.e., $v_k(A^t_k\cup\{g\})-v_k(A^t_k)=1$) and ({\rm b}) agent $k$'s subsidy (under $p^t=(p^t_1, \ldots, p^t_n)$) is at least as large as that of any other agent ($p^t_k \geq p^t_j$ for all $j$). In such a case, the allocation $\A^{t+1}$ obtained by assigning the good $g$ to the agent $k$ is envy-freeable, since the social welfare of $\A^{t+1}$ is one more than the social welfare of $\A^t$ and this is the maximum possible among all the reassignments of the bundles in $\A^{t+1}$ (see condition $(ii)$ in Theorem~\ref{thm:hs_ef}). Furthermore, for the allocation $\A^{t+1}$, the accompanying subsidies remain in $\{0,1\}$. This essentially follows from the facts that $k$ was one of the most subsidized agent and, when we assign the good $g$ to agent $k$, only the weights of the edges incident on $k$ (in the updated envy graph) change. In particular, the weight of the edges going out of $k$ necessarily decrease by $1$ and for each incoming edge the weight increases by at most $1$.\footnote{Note that under dichotomous valuations the edge weights in the envy graph are integers.} Using these observations one can show that the weights of all paths remain below $1$ and, hence (via Theorem \ref{thm:hs_sub}), the subsidies remain in $\{0,1\}$.

Even if there does not exist an agent $k$ that directly satisfies the above-mentioned conditions ({\rm a}) and ({\rm b}), we can still try to first reassign the current bundles (among the agents) and then look for such an agent. Specifically, the case detailed above continues to be relevant if there exists a permutation $\sigma$ such that allocation $\mathcal{B} = \A^t_\sigma$ is envy-freeable,\footnote{Here, the envy-freeability of $\mathcal{B}$ ensures that the accompanying subsidies are still either $0$ or $1$; see Lemma~\ref{lem:reshuffle}.} and, under $\mathcal{B}$, there exists an agent $\kappa$ that satisfies conditions ({\rm a}) and ({\rm b}). Encapsulating this case,\footnote{Note that with $\sigma$ as the identity permutation, one can capture the setting wherein a reassignment of the bundles is not required.} we will say that the current allocation $\A^t$ is \emph{extendable}, with good $g$, iff the desired permutation $\sigma$ and agent $\kappa$ exist; see Definition \ref{def:extend}. 

Section \ref{sec:extend} develops a subroutine, \textsc{Extend} (Algorithm \ref{Alg:Extend}), that efficiently identifies whether the current allocation $\A^t$ is extendable. If the allocation is extendable, then the subroutine returns the relevant permutation $\sigma$ and the agent $\kappa$. Using this subroutine, our algorithm \textsc{Alg} addresses the case of extendable solutions (in Lines \ref{line:if-extendable} to \ref{line:update-extend}). In particular, \textsc{Alg} allocates good $g$ and updates the envy-free solution $(\A^t, p^t)$ to $(\A^{t+1}, p^{t+1})$, while maintaining the invariant that the subsidies are either $0$ or $1$.

\floatname{algorithm}{Algorithm}
\begin{algorithm}[h!]
	\caption{\textsc{Alg}}\label{Alg:BinSubsidy} 
	\textbf{Input:} Instance $\instance{[n], [m], \val}$ with value oracle access to dichotomous valuations $v_i$s. \\
	\textbf{Output:} Allocation $\A$ and subsidy vector $p$.
	\begin{algorithmic}[1]
		\STATE Initialize index $t=1$ along with bundles $A^t_1= \ldots = A^t_n = \emptyset$ and subsidies $p^t_1 = \ldots = p^t_n = 0$. \label{step:initialize} 
		\COMMENT{Here, allocation $\mathcal{A}^t = (\emptyset, \ldots, \emptyset)$ and subsidy vector $p^t = (0, \ldots, 0)$.}
		\WHILE {$[m] \setminus \left( \cup_{i\in [n]} A^t_i \right) \neq \emptyset$}\label{goods_remain}
		\STATE Select any unallocated good $g \in [m] \setminus \left( \cup_{i\in [n]} A^t_i \right)$.\label{pick_good}
		 \IF{$(\A^t,p^t)$ is extendable with good $g$} \label{line:if-extendable}
		 \STATE Set $(\sigma,\kappa) \coloneqq \textsc{Extend} (\A^t,p^t,g)$ and set allocation $(B_1, \ldots, B_n) \coloneqq \A^t_\sigma$. \label{t+1_extend} \\ \COMMENT{$(B_1, \ldots, B_n)$ is obtained by reassigning bundles in $\mathcal{A}^t$, according to permutation $\sigma$.}
		  
		 \STATE Set allocation $\A^{t+1} \coloneqq (B_1,\ldots,B_{\kappa}\cup \{g\},\ldots, B_n)$. \label{line:update-extend} \\ \COMMENT{Good $g$ is included in agent $\kappa$'s bundle $B_\kappa$.} 
		 \ELSE
		 \STATE Let agent $s \coloneqq \textsc{FindSink}(\A^t,p^t,g)$. \label{line:find-sink}
		 \STATE Set allocation $\A^{t+1} \coloneqq (A^t_1,\ldots,A^t_s \cup \{g\},\ldots,A^t_n)$. \label{line:sink-assign} \COMMENT{Good $g$ is included in agent $s$'s bundle $A^t_s$.}
		\ENDIF
		\STATE Compute subsidy vector $p^{t+1}$ for the allocation $\A^{t+1}$. \COMMENT{We will prove that, in all cases, $\A^{t+1}$ is envy-freeable and invoke Theorem \ref{thm:hs_sub} to compute the subsidies.} \label{step:payment}
		\STATE Update $t\leftarrow t+1$			
		\ENDWHILE \label{step:endfor}
		\RETURN allocation $\A^t=(A^t_1, \ldots, A^t_n)$ and subsidy vector $p^t=(p_1^t,\ldots,p^t_n)$.
	\end{algorithmic}
\end{algorithm}

In the complementary case, wherein the current allocation $\A^t$ is not extendable, \textsc{Alg} (in Line \ref{line:find-sink}) invokes the subroutine \textsc{FindSink} (detailed in Section~\ref{sec:non-extend}). A key technical insight here is that, under dichotomous valuations and for an allocation $\A^t$ that is \emph{not} extendable with a good $g$, there necessarily exists an agent $s$ such that assigning $g$ to $s$ maintains the subsidies in $\{0,1\}$.\footnote{For agent $s$, the first condition of extendability might not hold.} The subroutine \textsc{FindSink} directly finds such an agent by trying to assign $g$ to different agents iteratively and checking whether subsidies remain in $\{0,1\}$. It is relevant to note that, while the subroutine is simple in design, its analysis requires intricate existential arguments. In particular, Section~\ref{sec:non-extend} establishes that \textsc{FindSink} finds such an agent $s \in [n]$ in polynomial time and, hence, provides a constructive proof of existence of the desired agent $s$.    

Therefore, in every possible case, our algorithm \textsc{Alg} assigns a good and updates the allocation, all the while maintaining the invariant that the allocation in hand is envy-freeable and requires subsidies that are either $0$ or $1$. 

For the analysis of the algorithm, we will use the following two propositions, which hold for dichotomous valuations. The proofs of the propositions are deferred to Appendix \ref{appendix:prop-proofs}. 

\begin{restatable}{proposition}{PropSW}
\label{lem:sw_increase}
Let $\mathcal{Y}=(Y_1, \ldots, Y_n)$ be an envy-freeable (partial) allocation. Also, assume that, for an agent $x \in [n]$ and an unallocated good $g$, we have $v_x(Y_x \cup \{g \}) - v_x(Y_x) = 1$. Then, the allocation $(Y_1,\ldots, Y_x\cup \{g\}, \ldots, Y_n)$ is envy-freeable as well. 
\end{restatable}


\begin{restatable}{proposition}{PropWeight}
\label{lem:weight_increase}
For any envy-freeable allocation $\mathcal{Y}=(Y_1, \ldots, Y_n)$, any agent $x$, and any (unallocated) good $g$, let allocation $\mathcal{Z} \coloneqq (Y_1, \ldots, Y_x\cup \{g\},\ldots,Y_n)$. Then, in the envy graph $G_{\mathcal{Z}}$, the weights of all the edges---except the ones incident on $x$---are the same as in $G_\mathcal{Y}$, i.e., for all edges $(i,j)$, with $i, j \in [n] \setminus \{x\}$, we have $w_\mathcal{Z}(i,j) = w_\mathcal{Y}(i,j)$. Furthermore, the weights of edges $(x,j)$ going out of $x$ satisfy $w_\mathcal{Z}(x, j) \leq w_\mathcal{Y}(x,j) $, and the weights of edges $(i,x)$ coming into $x$ satisfy $w_\mathcal{Z}(i,x) \leq w_\mathcal{Y}(i, x) + 1$.
\end{restatable}

For allocation $\mathcal{Z} \coloneqq (Y_1, \ldots, Y_x\cup \{g\},\ldots,Y_n)$, Proposition \ref{lem:weight_increase} implies that the weight of any path\footnote{Note that, by definition, a path joins \emph{distinct} vertices.} $P$ in the envy graph $G_{\mathcal{Z}}$ is at most one more than the weight of $P$ in the envy graph $G_\mathcal{Y}$, i.e., $w_{\mathcal{Z}}(P) \leq w_{\mathcal{Y}}(P)+1$. 

\subsection{Extendable Solutions}\label{sec:extend}
As mentioned previously, a relevant case for our algorithm occurs when, for an unallocated good $g$, there exists an agent $\kappa$ such that ({\rm a}) the marginal value of good $g$ for agent $\kappa$ is equal to one and ({\rm b}) agent $\kappa$ is one of the most subsidized agents. The notion of extendability (defined below) generalizes this case. Recall that  $\M(p)$---for any subsidy vector $p=(p_1, \ldots, p_n)$---denotes the set of the most subsidized agents, $\M(p) \coloneqq \left\{ i\in [n] \mid  p_i \geq p_j  \text{ for all } j\in [n] \right\}$. 

\begin{definition}[Extendable Solutions] \label{def:extend}
	An envy-free solution $ (\mathcal{A},p) $ is said to be extendable with good $g\in [m] \setminus \left( \cup_{i\in [n]}A_i \right)$ iff there exists a permutation $\sigma$ (over $[n]$) such that: 
	\begin{enumerate}
		\item [(i)] Allocation $\mathcal{B} \coloneqq \mathcal{A}_\sigma$ and subsidy vector $q \coloneqq p_\sigma$ constitute an envy-free solution $(\mathcal{B},q)$, and 
		\item [(ii)] There exists an agent $\kappa \in \M(q)$ with the property that $v_{\kappa}(B_{\kappa}\cup \{ g \}) - v_{\kappa}(B_{\kappa}) = 1$.
	\end{enumerate}
\end{definition}

The following lemma shows that, for an extendable solution $(\A, p)$, we can allocate the good $g$ and still maintain per-agent subsidy to be either $0$ or $1$. 

\begin{lemma}\label{lem:extendabillity}
In a fair division instance with dichotomous valuations, let $(\A, p) $ be an envy-free solution extendable with good $g$ and permutation $\sigma$. Also, assume that the subsidies $p \in \{0,1\}^n$. Then, the allocation $(B_1,\ldots,B_\kappa \cup \{g\},\ldots,B_n)$ is also envy-freeable with a subsidy of either $0$ or $1$ for each agent. Here, allocation $(B_1, \ldots, B_n) \coloneqq A_\sigma$ and $\kappa$ is the agent identified in the extendability criteria (Definition \ref{def:extend}).
\end{lemma}

\begin{proof}
Since $(\A,p)$ is extendable (with good $g$ and permutation $\sigma$), by definition of extendability, the allocation $\mathcal{B} = (B_1, \ldots, B_n) = \A_\sigma$ is envy-freeable,  via subsidy vector $q = p_\sigma$. Also, $v_{\kappa}(B_{\kappa}\cup \{{g}\})-v_\kappa(B_\kappa)=1$ for agent $\kappa \in \M(q)$. Therefore, invoking  Proposition \ref{lem:sw_increase} (with allocation $\mathcal{Y} = \mathcal{B}$ and agent $x = \kappa$), we get that allocation $\mathcal{C} \coloneqq (B_1,\ldots,B_\kappa \cup \{g\},\ldots,B_n)$ is envy-freeable.  

We will complete the proof below by showing that for $\mathcal{C}$ the required subsidies are either $0$ or $1$. Towards this, we consider two complementary and exhaustive cases: either $q_\kappa = 0$ or $q_\kappa =1$; recall that, under the lemma assumption, the subsidy vector $p \in \{0,1\}^n$ and, hence, $q=p_\sigma \in \{0,1\}^n$.   

First, consider the case wherein $q_\kappa =0$. The fact that $\kappa \in \M(q)$ (i.e., $\kappa$ is among the most subsidized agents under $q$) implies $q_i=0$ for all agents $i \in [n]$. Hence, in the envy graph $G_\mathcal{B}$ and starting from any agent $i \in [n]$, the weight of any path is at most $0$; see Theorem \ref{thm:hs_sub} and recall that $(\mathcal{B}, q)$ is an envy-free solution. Using Proposition~\ref{lem:weight_increase} (with $\mathcal{Y} = \mathcal{B}$, $\mathcal{Z} = \mathcal{C}$, and agent $x = \kappa$), we get that in $G_\mathcal{C}$ and for any agent $i\in [n]$, the weight of any path starting at $i$ is at most $1$. Hence, Theorem \ref{thm:hs_sub} implies that, for the (envy-freeable) allocation $\mathcal{C}$, a per-agent subsidy of either $0$ or $1$ suffices to realize an envy-free solution.

Next, consider the case $q_\kappa=1$. Here, we explicitly define a subsidy vector $\widehat{q} = (\widehat{q}_1, \ldots, \widehat{q}_n) \in \{0,1\}^n$ and show that $(\mathcal{C}, \widehat{q})$ is an envy-free solution. In particular, set $\widehat{q}_\kappa = 0$ and $\widehat{q}_i = q_i$ for all agents $i \neq \kappa$. Note that for all agents $i \neq \kappa$, the bundle $C_i=B_i$ and subsidy $\widehat{q}_i=q_i$. Hence, the fact that $(\mathcal{B}, q)$ is an envy-free solution implies $v_i(C_i) + \widehat{q}_i \geq v_i(C_j) + \widehat{q}_j$ for all $i, j \in [n] \setminus \{ \kappa \}$. In addition, agent $\kappa$ does not envy any other agent $i \in [n]$:
\begin{align*}
v_\kappa(C_\kappa) + \widehat{q}_\kappa = v_\kappa(B_\kappa \cup \{g\}) + 0 = v_\kappa(B_\kappa) + 1 = v_\kappa(B_\kappa) + q_\kappa \geq v_\kappa(B_i) + q_i. 
\end{align*}   
The envy-freeness towards $\kappa$ is also maintained for all agents $i \in [n] \setminus \{ \kappa \}$:
\begin{align*}
v_i(C_i) + \widehat{q}_i = v_i(B_i) + q_i \geq v_i(B_\kappa) + q_\kappa \geq v_i(B_\kappa \cup \{ g \}) - 1 + q_k = v_i(C_\kappa) + \widehat{q}_\kappa.
\end{align*}
The last inequality follows from the fact that $v_i$ is dichotomous. These bounds together imply that $(\mathcal{C}, \widehat{q})$ is an envy-free solution and, by construction, $\widehat{q} \in \{0,1\}^n$. Therefore, even in the current case ($q_\kappa = 1$), for allocation $\mathcal{C} = (B_1,\ldots,B_\kappa \cup \{g\},\ldots,B_n)$ the required subsidies are either $0$ or $1$. The lemma stands proved. 
\end{proof}

\floatname{algorithm}{Algorithm}
\begin{algorithm}[h!]
	\caption{\textsc{Extend}}\label{Alg:Extend} 
	\textbf{Input:} Instance $\instance{[n], [m], \val}$, an envy-free solution $(\A,p)$, and a good $g\in[m]$. \\
	\textbf{Output:} A permutation and an agent that extend the input $(\A,p)$ with good $g$, or return that it is not extendable.
	\begin{algorithmic}[1]
		\FOR{each $k \in [n]$ and $\ell \in \M(p)$ with the property that $v_k(A_\ell \cup \{g \}) - v_k(A_\ell) = 1$} \label{line:select-kl}
		\STATE Consider a complete bipartite graph $H$ between sets $[n] \setminus \{k\} $ and $[n] \setminus \{ \ell \}$.  For each edge $(i,j)$ in $H$ set the weight to be $v_i(A_j)$. \label{line:def-bipartite}
		\STATE Compute a maximum-weight matching $\rho$ in $H$. \label{G'}
		\STATE Set $\rho(k) = \ell$. \COMMENT{$\rho$ is a permutation over $[n]$.} \label{line:complete-rho}
		\STATE If $\sum_{i=1}^n v_i(A_{\rho(i)}) \geq \sum_{i=1}^n v_i(A_i)$, then return $(\rho,k)$. \label{line:return-ex}
		\ENDFOR
		\RETURN ``$(\A,p)$ is not extendable with good $g$.''
	\end{algorithmic}
\end{algorithm}

The following lemma shows that the subroutine \textsc{Extend} efficiently identifies whether a given solution $(\A, p)$ is extendable, i.e., the subroutine efficiently tests whether Definition \ref{def:extend} holds, or not. The proof of the lemma appears in Appendix \ref{appendix:extend}.

\begin{restatable}{lemma}{LemmaExtendWorks}
\label{lem:extend_correct}
In a fair division instance with dichotomous valuations, let $(\mathcal{A},p)$ be an envy-free solution. 
\begin{itemize}
	\item If $(\A,p)$ is extendable with good $g \in [m]\setminus \left( \cup_{i=1}^n A_i \right)$, then the \textsc{Extend} subroutine returns a permutation $\sigma$ and an agent $\kappa$ that satisfy the extendability criteria (Definition \ref{def:extend}). 
	\item Otherwise, if $(\A,p)$ is not extendable with good $g$, then subroutine \textsc{Extend} correctly reports as such.
\end{itemize}
\end{restatable}

\subsection{Non-Extendable Solutions}\label{sec:non-extend}
This section addresses the case wherein the maintained allocation is not extendable. The subroutine \textsc{FindSink} takes as input an envy-free solution $(\A,p)$---with subsidy vector $p \in \{0,1\}^n$---and a good $g$. The subroutine tentatively assigns the good to one agent at a time and checks if the subsidies remain in $\{0,1\}$. We will show that if the input $(\A,p)$ is not extendable, then \textsc{FindSink} necessarily succeeds in finding such an agent (Lemmas \ref{lem:fs_terminate} and \ref{lem:fs_correct}).   

\floatname{algorithm}{Algorithm}
\begin{algorithm}[h!]
	\caption{\textsc{FindSink}}\label{Alg:FindSink}  
	\textbf{Input:} Instance $\instance{[n], [m], \val}$, an envy-free solution $(\A,p)$, and a good $g\in[m]$. \\
	\textbf{Output:} An agent $s \in [n]$.
	\begin{algorithmic}[1]
		\STATE Select an arbitrary agent $s \in \M(p)$ and set allocation $\mathcal{X} = (A_1, A_2, \ldots, A_{s} \cup \{g\}, \ldots, A_n)$. \label{line:find-sink-initialize} \\ \COMMENT{We will show that---for non-extendable inputs---allocation $\mathcal{X}$ is envy-freeable.} 
		\STATE Compute the subsidy vector $(\varphi_1, \ldots, \varphi_n)$ for allocation $\mathcal{X}$. \COMMENT{See Theorem~\ref{thm:hs_sub}.}
		\WHILE{there exists an agent $j$ with subsidy $\varphi_j\geq 2$}\label{loop_FS} 
				\STATE Update $s \leftarrow j$ \label{agent_j} and update $\mathcal{X} = (A_1, A_2, \ldots, A_{s} \cup \{g\},\ldots, A_n)$. \label{line:new-candidate}
				\STATE Compute the subsidy vector $(\varphi_1, \ldots, \varphi_n)$ for allocation $\mathcal{X}$.
		\ENDWHILE
		\RETURN agent $s$.
	\end{algorithmic}
\end{algorithm}

The following proposition shows that all the agents $s$ considered by \textsc{FindSink} are among the most subsidized ones, under $p$, and the considered allocations $\mathcal{X} \coloneqq (A_1,\ldots,A_s\cup\{g\},\ldots,A_n)$  are envy-freeable. Note that, throughout the execution of the subroutine, $(\A,p)$ remains unchanged.

\begin{restatable}{lemma}{LemmaFindSinkInduct}
\label{lem:s}
In a fair division instance with dichotomous valuations, let $(\A,p) $ be an envy-free solution with subsidies $p \in \{0,1\}^n$. If input $(\A,p)$ is not extendable with the given good $g$, then each agent $s$ considered in \textsc{FindSink} (Line \ref{line:new-candidate}) belongs to the set $\M(p)$  and each considered allocation $\mathcal{X} \coloneqq (A_1,\ldots,A_s\cup\{g\},\ldots,A_n)$  is envy-freeable.
\end{restatable}

\noindent
The proof of the Lemma is deferred to Appendix \ref{appendix:prop-non-ext}.

The following lemma is a key technical result for $\textsc{FindSink}$. It shows that the subroutine terminates in polynomial time, when the input $(\mathcal{A},p)$ is not extendable. 

\begin{lemma}\label{lem:fs_terminate}
In a fair division instance with dichotomous valuations, let $(\A,p) $ be an envy-free solution with subsidies $p \in \{0,1\}^n$. Also, assume that $(\A,p)$ is not extendable with good $g$. Then, given $(\A,p)$ and $g$ as input, the subroutine \textsc{FindSink} terminates in polynomial time.
\end{lemma}
\begin{proof}
We will prove that no agent $j\in [n]$ is selected more than once in the while-loop (Line \ref{loop_FS}) of \textsc{FindSink}. This will imply that the subroutine iterates at most $n$ times and, hence, will help establish the subroutine's time complexity. 

Assume, for a contradiction, that agent $j \in [n]$ is selected more than once. In particular, let $s^0, s^1, \ldots, s^T$ denote a sequence of agents selected (one after the other) in \textsc{FindSink} such that $s^0 = s^T = j$. We will first show that, for each $\tau \in [T]$, there exists a path $P^\tau$ between $s^\tau$ and $s^{\tau-1}$ with weight $w_\A (P^\tau)  \geq 0$. Note that the weight considered here is in the envy graph $G_\A$ and, throughout the execution of \textsc{FindSink}, the allocation $\A$ remains unchanged. 

For each $\tau \in [T]$, write allocation $\mathcal{X}^\tau \coloneqq (A_1, \ldots, A_{s^\tau} \cup \{g\}, \ldots, A_n)$. Note that the agent $s^\tau$ was selected right after $s^{\tau-1}$ and, hence, agent $s^\tau$ must have received a subsidy of at least $2$ under the subsidy vector computed for allocation $\mathcal{X}^{\tau-1}$. Therefore, in the graph $G_{\mathcal{X}^{\tau-1}}$ and starting at $s^\tau$, there exists a path ${Q}^\tau$ with weight 
\begin{align}
w_{\mathcal{X}^{\tau -1}} ({Q}^\tau) \geq 2 \label{ineq:defQ}
\end{align}
Here, we can apply Theorem \ref{thm:hs_sub}, since allocation $\mathcal{X}^{\tau-1}$ is envy-freeable (Lemma \ref{lem:s}). Also, the path ${Q}^\tau$ must pass through $s^{\tau-1}$, otherwise its weight remains the same as in $G_\A$ (Proposition \ref{lem:weight_increase}), i.e., $w_{\mathcal{X}^{\tau-1}} ({Q}^\tau)\leq 1$ (which would contradict (\ref{ineq:defQ})). We will set $P^\tau$ to be the subpath---of ${Q}^\tau$---that goes from $s^\tau$ to $s^{\tau-1}$. The claim below lower bounds the weight of $P^\tau$ in the envy graph $G_\A$. 

\begin{claim} \label{claim:path-nonzero}
$w_\A(P^\tau) \geq 0$ for all $1 \leq \tau \leq T$. 
\end{claim}
\begin{proof}
The weight of the path $Q^\tau$ in the graph $G_\A$ satisfies 
\begin{align}
w_\A(Q^\tau) \leq w_\A(P^\tau) +1 \label{ineq:QP}
\end{align}
This inequality follows from the observation that the subpath of $Q^\tau$ that appears after $P^\tau$ (i.e., the subpath of $Q^\tau$ that starts at $s^{\tau-1}$) is itself of weight at most one in the graph $G_\A$ (Theorem \ref{thm:hs_sub}); recall that $p \in \{0,1\}^n$. Furthermore, applying Proposition~\ref{lem:weight_increase} (with $\mathcal{Y} = \mathcal{A}$, $\mathcal{Z} = \mathcal{X}^{\tau-1}$, and agent $x = s^{\tau-1}$), we get $w_\A(Q^\tau) \geq w_{\mathcal{X}^{\tau-1}} (Q^\tau) - 1 \geq 1$; the last inequality follows from equation (\ref{ineq:defQ}). Therefore, inequality (\ref{ineq:QP}) reduces to $w_\A(P^\tau) \geq 0$. The claim stands proved. 
\end{proof}

Building upon Claim \ref{claim:path-nonzero}, we will next show that the collection of paths $P^\tau$s can be covered by a family of zero-weight cycles. 
Recall that $s^0, s^1, \ldots, s^T$ is the sequence of agents under consideration with $s^0 = s^T$. Also, for each $\tau \in [T]$, the path $P^\tau$ connects $s^\tau$ to $s^{\tau-1}$. Hence, the paths $P^\tau$s connected together form a closed directed walk, which we will refer to as $D$. These paths might have intersections, where an edge appears in multiple $P^\tau$s. For analysis, we replace these edges with multi-edges such that each (edge) copy appears exactly once in $D$. By Claim \ref{claim:path-nonzero}, we have $w_\A(P^\tau) \geq 0$, for each $\tau \in [T]$. Therefore, the weight of the directed walk $D$ (in $G_\A$) satisfies $w_{\A}(D) = \sum_{\tau =1}^T w_{\A}(P^\tau) \geq 0$. Also, by construction, in the directed walk $D$, the in-degree of each vertex is equal to its out-degree. Hence, $D$ can be decomposed into a family $\mathcal{K}$ of cycles. The family $\mathcal{K}$ satisfies 
\begin{align} 
\sum_{K \in \mathcal{K}} w_\A(K) = w_\A(D) \geq 0 \label{ineq:KD}
\end{align}
Here, $w_\A(K)$ denotes the weight of the cycle $K$ in $G_\A$. Recall that the allocation $\A$ is envy-freeable and, hence $G_\A$ does not contain any positive-weight cycle (Theorem~\ref{thm:hs_ef}). Using this observation and inequality (\ref{ineq:KD}) we get that 
\begin{align} 
w_\A(K) = 0 \qquad \text{for each cycle $K \in \mathcal{K}$} \label{eq:cyclefamily}
\end{align} 
Hence, the collection of paths $P^\tau$s can be covered by a family of zero-weight cycles $K \in \mathcal{K}$. 

Next, we will complete the proof of the lemma by deriving a contradiction that $(\A, p)$ is extendable. Towards this, we establish existence of a relevant cycle $\widehat{K} \in \mathcal{K}$.  
\begin{claim} \label{claim:cycle-zero}
There exists a cycle $\widehat{K} \in \mathcal{K}$ containing an edge $(k,s^0)$ such that $v_k(A_{s^0} \cup \{ g \}) - v_k(A_{s^0}) = 1$.  
\end{claim}
\begin{proof}
Consider the path $P^1$ (i.e., consider $\tau =1$) that connects $s^1$ to $s^0$. Let $(k,s^0)$ be the last edge in $P^1$. Since $P^1$ is a subpath of $Q^1$, the edge $(k,s^0)$ is contained in $Q^1$ as well. Furthermore, inequality (\ref{ineq:defQ}) with $\tau = 1$ gives us $w_{\mathcal{X}^0} (Q^1) \geq 2$. Also, using the fact that $p \in \{0,1\}^n$ and Theorem  \ref{thm:hs_sub} we have $w_\A(Q^1) \leq 1$. The increase in the weight of the path $Q^1$, as we move from $G_\A$ to $G_{\mathcal{X}^0}$, can be accounted only by an increase in the weight of the edge $(k,s^0)$; see Proposition \ref{lem:weight_increase}. Therefore, we have $v_k(A_{s^0} \cup \{ g \}) - v_k(A_{s^0})  = w_{\mathcal{X}^0} (k,s^0) - w_\A(k,s^0) = 1$. 

By construction, the cycles in the family $\mathcal{K}$ cover the collection of paths $P^\tau$s. Hence, there exists a cycle $\widehat{K} \in \mathcal{K}$ that contains the edge $(k,s^0)$. The claim holds for this cycle $\widehat{K}$. 
\end{proof}

Using the cycle $\widehat{K}$ identified in Claim \ref{claim:cycle-zero}, we will define a permutation $\widehat{\sigma}$ which will lead to the desired contradiction. The permutation $\widehat{\sigma}$ will essentially reassign the bundles in the reverse order of the cycle $\widehat{K} = i_1 \rightarrow i_2 \rightarrow \ldots \rightarrow i_h \rightarrow i_1$:  
for all agents $x$ not in the cycle (i.e., $ x \notin \{i_1, \ldots, i_h\}$) set $\widehat{\sigma}(x) = x$.  For all agents $i_a$ in $\widehat{K}$, set $\widehat{\sigma}(i_a)$ to be her successor in the cycle, i.e., $\widehat{\sigma}(i_a) = i_{a+1}$ for $1 \leq a < h$ and $\widehat{\sigma}(i_h) = i_1$.  

The following claim shows that the allocation obtained by reassigning bundles of $\A$ according to $\widehat{\sigma}$ is envy-freeable. Also, note that the edge $(k, s^0)$ belongs to the cycle $\widehat{K}$ (Claim \ref{claim:cycle-zero}). Hence, $\widehat{\sigma}(k) = s^0$ and agent $k$ receives the bundle $A_{s^0}$ in the allocation $\mathcal{B} = \A_{\widehat{\sigma}}$. 

\begin{claim}
\label{claim:hat-envy}
The allocation $\mathcal{B} = (B_1, \ldots, B_n) \coloneqq \A_{\widehat{\sigma}}$ is envy-freeable. 
\end{claim}
\begin{proof}
We will show that the social welfare of $\mathcal{B}  = \A_{\widehat{\sigma}}$ is same as that of the envy-freeable allocation $\A$. Hence, the claim will follow, via Theorem \ref{thm:hs_ef}. Recall that cycle $\widehat{K} \in \mathcal{K}$ and, hence, the weight of $\widehat{K}$ in $G_\A$ is equal to zero, $w_\A(\widehat{K}) = 0$; see equation (\ref{eq:cyclefamily}). In particular, for $\widehat{K} = i_1 \rightarrow i_2 \rightarrow \ldots \rightarrow i_h \rightarrow i_1$ we have 
\begin{align}
\sum_{a=1}^{h-1} w_\A(i_a, i_{a+1}) + w_\A(i_h, i_1) = 0 \label{eq:cyclewts}
\end{align} 

Since $(B_1, \ldots, B_n) = \A_{\widehat{\sigma}}$, the construction of $\widehat{\sigma}$ gives us $B_{i_a} = A_{i_{a+1}}$ for all $1 \leq a < h$ and $B_{i_h} = A_{i_1}$. Using this indexing along with equation (\ref{eq:cyclewts}) and the definition of $w_\A( \cdot, \cdot)$, we obtain $\sum_{a=1}^h \left( v_{i_a} (B_{i_a}) - v_{i_a} (A_{i_a}) \right) =0$. That is, 
\begin{align}
\sum_{a=1}^h v_{i_a} (B_{i_a}) = \sum_{a=1}^h v_{i_a} (A_{i_a}) \label{eq:poorindex}
\end{align} 
Also, by construction, we have $B_x = A_x$ for all agents $x$ not in the cycle. That is, $v_x(B_x) = v_x(A_x)$ for all $ x \notin \{i_1, \ldots, i_h\}$. Adding these equalities with equation (\ref{eq:poorindex}), we get the desired bound on the social welfare of ${\mathcal{B}} = (B_1, \ldots, B_n)$:
\begin{align}
\sum_{i=1}^n v_i(B_i) = \sum_{i=1}^n v_i(A_i) \label{eq:swhat}
\end{align}  
Therefore, allocation $\mathcal{B}$ is envy-freeable (Theorem \ref{thm:hs_ef}). This completes the proof of the claim. 
\end{proof}

Finally, we will derive the desired contradiction that $(\A,p)$ is extendable with good $g$: 
\begin{itemize}
\item[(i)] For the permutation $\widehat{\sigma}$ defined above, the allocation $\mathcal{B} = \A_{\widehat{\sigma}}$ is envy-freeable (Claim \ref{claim:hat-envy}). Hence, Lemma \ref{lem:reshuffle} implies that subsidy vector $q \coloneqq p_{\widehat{\sigma}}$ realizes an envy-free solution $(\mathcal{B}, q)$.  
\item[(ii)] As mentioned previously, for the agent $k$ identified in Claim \ref{claim:cycle-zero} we have $\widehat{\sigma}(k) =s^0$. Therefore, in allocation $\mathcal{B} = \A_{\widehat{\sigma}}$, agent $k$'s bundle $B_k = A_{s^0}$. We can, hence, express the equality from Claim \ref{claim:cycle-zero} as $v_k(B_k \cup \{g\}) - v_k(B_k) = 1$. Also, since $s^0 \in \M(p)$ (Lemma \ref{lem:s}), we have $k \in \M(q)$. 
\end{itemize}
Therefore, permutation $\widehat{\sigma}$ and agent $k$ satisfy extendability criteria (Definition \ref{def:extend}). This, however, contradicts the lemma assumption that the input $(\A,p)$ is not extendable. Therefore, the subroutine \textsc{FindSink} selects any agent $s^0$ at most once and it terminates in polynomial time. The lemma stands proved. 
\end{proof}

The subroutine \textsc{FindSink} is designed to terminate only when it has identified an agent $s$ to whom we can assign the good $g$ and maintain the subsidies to be in $\{0,1\}$. The next lemma formalizes this observation.  

\begin{lemma}\label{lem:fs_correct}
In a fair division instance with dichotomous valuations, let $(\A,p) $ be an envy-free solution with $p \in \{0,1\}^n$. Also, assume that $(\A,p)$ is not extendable with good $g$ and let $s \in [n]$ be the agent returned by \textsc{FindSink} (given $(\A,p)$ and $g$ as input). Then, the allocation $(A_1, \ldots, A_s \cup \{g \}, \ldots, A_n)$ is envy-freeable with a subsidy of either $0$ or $1$ for each agent. 
\end{lemma}
\begin{proof}
Since $s$ is one of the agents selected in \textsc{FindSink}, Lemma \ref{lem:s} implies that the allocation $(A_1, \ldots, A_s \cup \{g \}, \ldots, A_n)$ is envy-freeable. Furthermore, the termination condition of the while-loop in \textsc{FindSink} (Line \ref{loop_FS}) ensures that, for the allocation $(A_1, \ldots, A_s \cup \{g \}, \ldots, A_n)$, the required subsidies are less than $2$. Since the subsidies are nonnegative integers (for dichotomous valuations), they are either $0$ or $1$. This completes the proof. 
\end{proof}

\subsection{Proof of Theorem \ref{thm:main}}

In this section, we will prove Theorem~\ref{thm:main} by establishing that the solution returned by $\textsc{Alg}$ (Algorithm \ref{Alg:BinSubsidy}) is envy-free and requires, for each agent, a subsidy of either $0$ or $1$. 
\mainthm*
\begin{proof}
We will show that, with iteration counter $t$, the algorithm inductively maintains an envy-free solution $(\A^t, p^t)$ wherein $\A^t=(A^t_1, \ldots A^t_n)$ is a (partial) envy-freeable allocation and the subsidy vector $p^t \in \{0,1\}^n$. The initialization in Line \ref{step:initialize} of \textsc{Alg} ensures that this property holds for $t=0$, since the empty allocation $\A^0$ is envy-freeable with zero subsidies. 

Next, consider any iteration $t\geq 1$ of the while-loop (Line \ref{goods_remain}) in \textsc{Alg} and let $g$ be the good selected in Line \ref{pick_good}. For the current solution $(\A^t, p^t)$, there are two complementary and exhaustive cases: \\

\noindent 
{\it Case {\rm I}: $(\A^t, p^t)$ is extendable with good $g$.} Here, Lemma~\ref{lem:extend_correct} implies that the permutation $\sigma$ and the agent $\kappa$ returned by the \textsc{Extend} subroutine satisfies the extendability criteria. Hence, by Lemma~\ref{lem:extendabillity}, we get that assigning $g$ to agent $\kappa$'s bundle---as in Lines \ref{t+1_extend} and \ref{line:update-extend}---leads to an envy-free solution $(\A^{t+1},p^{t+1})$ with $p^{t+1}\in \{0,1\}^n$. 

\noindent 
{\it Case {\rm II}: $(\A^t, p^t)$ is not extendable with good $g$.} This case will be correctly identified by \textsc{Alg} (Lemma~\ref{lem:extend_correct}) and the algorithm will then invoke the \textsc{FindSink} subroutine (Line \ref{line:find-sink}). The subroutine in turn will find (in polynomial time) an agent $s \in [n]$ (Lemma \ref{lem:fs_terminate}) such that assigning the good $g$ to agent $s$---as in Line \ref{line:sink-assign}---provides an envy-free solution $(\A^{t+1},p^{t+1})$ with $p^{t+1}\in \{0,1\}^n$ (Lemma~\ref{lem:fs_correct}).

Therefore, in both cases, the updated solution $(\A^{t+1}, p^{t+1})$ is envy-free with $p^{t+1} \in \{0,1\}^n$. This, overall, shows that the returned solution $(\A,p)$ satisfies the desired properties. 

Furthermore, note that the subroutines \textsc{Extend} and \textsc{FindSink} run in polynomial time, and only require value queries. This establishes the efficiency of \textsc{Alg} and completes the proof.  
\end{proof}

%% file: conclusion.tex
\section{Conclusion and Future Work}
We prove that, under dichotomous valuations, envy-freeness can always be achieved with a subsidy of at most $1$ per agent. This bound is tight and our proof is constructive. Specifically, our algorithm assigns the goods iteratively while maintaining envy-freeness, with bounded subsidies. Even though, at a high level, our algorithm might seem like a refinement of the method used for finding $\EFone$ allocations \cite{LMMS042}, the two approaches are distinct. In particular, the current algorithm resolves positive-weight cycles (in the envy graph) and, by contrast, finding $\EFone$ allocations (under monotone valuations) entails resolution of top trading cycles \cite{LMMS042}. We can, in fact, construct an instance wherein a specific execution of the developed algorithm returns an allocation that is not $\EFone$; see Appendix \ref{appendix:efone}. Hence, extending the current work to additionally obtain the $\EFone$ guarantee is a relevant direction of future work. Another interesting direction would be to obtain tight subsidy bounds for general monotone valuations.   

%% file: appendix.tex
\appendix

\section{Proof of Lemma \ref{lem:reshuffle}}
\label{appendix:lemma-shuffle}
This section restates and proves Lemma \ref{lem:reshuffle}. 

\LemmaShuffle*
\begin{proof}
We will prove that allocation $\mathcal{B} = (B_1, \ldots, B_n) = \A_\sigma$ satisfies $v_i(B_i)+p_{\sigma(i)}\geq v_i(B_j)+p_{\sigma(j)}$, for all agents $i, j \in [n]$. This will establish the lemma. Given that $(\A,p)$ is an envy-free solution, we have, for each agent $i \in [n]$:
\begin{align}\label{eq_1}
v_i(A_i)+p_i & \geq  v_i(A_{\sigma(i)})+p_{\sigma(i)}
\end{align}
Summing (\ref{eq_1}) over all $i\in [n]$ gives us 
\begin{align}\label{eq_2}
		\sum_{i=1}^n v_i(A_i) + \sum_{i=1}^n p_i & \geq  \sum_{i=1}^n v_i(A_{\sigma(i)}) + \sum_{i=1}^n p_{\sigma(i)} 
\end{align}
Furthermore, the lemma assumption ensures that permutation $\sigma$ maintains envy-freeability, i.e., both allocations $\A$ and $\mathcal{B} = \A_\sigma$ are envy-freeable. Hence, using condition (ii) of Theorem \ref{thm:hs_ef}, we obtain $\sum_{i=1}^n v_i(A_i) = \sum_{i=1}^n v_i(A_{\sigma(i)})$. In addition, $\sum_{i=1}^n p_i=\sum_{i=1}^n p_{\sigma(i)}$; recall that $\sigma$ is a permutation. These observations imply that inequality (\ref{eq_2}) in fact holds with an equality. Therefore, for any agent $i \in [n]$, inequality (\ref{eq_1}) cannot be strict. That is, the following equality holds for every agent $i \in [n]$:
\begin{align}
\label{eq_3}
v_i(A_i)+p_i =  v_i(A_{\sigma(i)})+p_{\sigma(i)}
\end{align}

Now, to show that $(\mathcal{B}, p_\sigma)$ is envy-free, consider any agents $i, j \in [n]$ and note that the definition $\mathcal{B}$ gives us
\begin{align*}
v_i(B_i)+p_{\sigma(i)}&=v_i(A_{\sigma(i)})+p_{\sigma(i)} \\
& = v_i(A_i)+p_{i} \tag{via equation (\ref{eq_3})} \\
& \geq v_i(A_{\sigma(j)})+p_{\sigma(j)} \tag{since $(\A, p)$ is envy-free} \\
& =v_i(B_j)+p_{\sigma(j)} \tag{since $\mathcal{B} = \A_\sigma$}  
\end{align*}
This establishes the envy-freeness of $(\mathcal{B}, p_\sigma)$ and completes the proof. 
\end{proof}

\section{Missing Proofs from Section \ref{section:main-algorithm}}
\label{appendix:prop-proofs}
Here, we restate and prove Propositions \ref{lem:sw_increase} and \ref{lem:weight_increase}.

\PropSW*
\begin{proof}  We will show that allocation $\mathcal{Z} \coloneqq (Y_1,\ldots, Y_x\cup \{g\}, \ldots, Y_n)$ satisfies condition $(ii)$ of Theorem \ref{thm:hs_ef} and, hence, is envy-freeable. Towards this, consider any permutation $\sigma$ over $[n]$ and write allocation $\widehat{\mathcal{Z}} \coloneqq \mathcal{Z}_\sigma$. Allocation $\widehat{\mathcal{Z}}$ is obtained by reassigning the bundles in $\mathcal{Z}=(Y_1,\ldots, Y_x\cup \{g\}, \ldots, Y_n)$. Hence, there exists an agent $a$ whose bundle $\widehat{Z}_a = Y_x \cup \{g\}$. All the other bundles in $\widehat{\mathcal{Z}}$ are identical to ones in the allocation $\mathcal{Y}$. Therefore, $\mathcal{Y}_\sigma = (\widehat{Z}_1, \ldots, \widehat{Z}_a\setminus \{g\}, \ldots, \widehat{Z}_n)$. Given that $\mathcal{Y}$ is envy-freeable, it achieves the maximum possible social welfare across all possible reassignments of its bundles (Theorem \ref{thm:hs_ef}). Hence, $ \sum_{i \neq a} v_i(\widehat{Z}_i) \ + v_a(\widehat{Z}_a \setminus \{ g \}) \leq \sum_{i=1}^n v_i(Y_i)$. This inequality along with the fact that the agents' valuations are dichotomous gives us 
\begin{align}
\sum_{i=1}^n v_i(\widehat{Z}_i) \leq \sum_{i=1}^n v_i(Y_i) + 1 \label{ineq:ztoy}
\end{align} 

Furthermore, via the proposition assumption, we have $v_x(Y_x \cup \{g \}) - v_x(Y_x) = 1$, for agent $x$ and good $g$. Therefore, the social welfare of allocation $\mathcal{Z}=(Y_1,\ldots, Y_x\cup \{g\}, \ldots, Y_n)$ satisfies $\sum_{i=1}^n v_i(Z_i) = \sum_{i=1}^n v_i(Y_i) + 1$. Using this equation and inequality (\ref{ineq:ztoy}), we obtain $\sum_{i=1}^n v_i(Z_i) \geq \sum_{i=1}^n v_i(\widehat{Z}_i)$. This, overall, shows that the the social welfare of $\mathcal{Z}$ is at least as high as that $\mathcal{Z}_\sigma$, for any permutation $\sigma$.  Therefore, allocation $\mathcal{Z}$ is envy-freeable (Theorem \ref{thm:hs_ef}). This completes the proof. 
\end{proof}

\PropWeight*
\begin{proof}
For each agent $i \neq x$, the bundle agent $i$ receives in allocation $\mathcal{Z} = (Y_1, \ldots, Y_x\cup \{g\},\ldots,Y_n)$ is same as her bundle in allocation $\mathcal{Y}=(Y_1, \ldots, Y_n)$. Therefore, in the envy graph $G_{\mathcal{Z}}$, the weight of any edge $(i,j)$---with $i, j \in [n] \setminus \{ x\}$---is the same as its weight in $G_\mathcal{Y}$; specifically, $w_\mathcal{Z}(i,j) = v_i(Z_j) - v_i(Z_i) = v_i(Y_j) - v_i(Y_i) = w_\mathcal{Y}(i,j)$.  

Furthermore, the weights of the edges going out of $x$ do not increase. In particular, for any edge $(x, j)$ we have $w_\mathcal{Z}(x, j) = v_x(Z_j) - v_x(Z_x) = v_x(Y_j) - v_x(Y_x \cup \{g\}) \leq v_x(Y_j) - v_x(Y_x) = w_\mathcal{Y}(x, j)$; the last inequality follows from the fact that the valuation $v_x$ is monotonic. 

Finally, we note that the weights of the incoming edges $(i,x)$ increase by at most one: $w_\mathcal{Z}(i, x) = v_i(Z_x) - v_i(Z_i) = v_i(Y_x \cup \{g\}) - v_i(Y_i) \leq v_i(Y_x) + 1 -  v_i(Y_i) = w_\mathcal{Y}(i,x) + 1$. Here, for the last inequality we use the fact that $v_i$ is dichotomous. 
\end{proof}

\subsection{Proof of Lemma \ref{lem:extend_correct}}
\label{appendix:extend}

We restate and prove Lemma \ref{lem:extend_correct} in this section. 

\LemmaExtendWorks*
\begin{proof} 
We first show that if the subroutine returns a permutation $\rho$ and an agent $k \in [n]$ (in Line \ref{line:return-ex}), then indeed the given solution $(\A, p)$ is extendable with good $g$. Write allocation $\mathcal{B} \coloneqq \A_\rho$. The successful execution of the if-condition in Line \ref{line:return-ex} implies that (analogous to the envy-freeable allocation $\A$) allocation $\mathcal{B}$ maximizes the social welfare across all reassignments of its bundles among the agents (Theorem \ref{thm:hs_ef}). Hence, $\mathcal{B}$ is envy-freeable as well. Furthermore, with subsidy vector $q \coloneqq p_\rho$ we obtain an envy-free solution $(\mathcal{B},q)$; see Lemma \ref{lem:reshuffle}. Therefore, condition $(i)$ of Definition \ref{def:extend} holds for the returned permutation $\rho$. In addition, for the returned agent $k$, we have $\rho(k) = \ell \in \M(p)$ and $v_k(A_\ell \cup \{g \}) - v_k(A_\ell) =1$; see Line \ref{line:complete-rho} and Line \ref{line:select-kl}. Since $\mathcal{B} = (B_1, \ldots, B_n) = \A_\rho$ and $q = p_\rho$, we get that $k \in \M(q)$ and $v_k(B_k \cup \{ g \}) - v_k(B_k)=1$. That is, condition $(ii)$ of Definition \ref{def:extend} holds as well. Hence, the returned permutation and agent witness the extendability of the input.   

Next, we complete the proof by establishing that if $(\A,p)$ is extendable with good $g$, then $\textsc{Extend}$ necessarily returns a permutation along with an agent. Since $(\A, p)$ is extendable, there exist a permutation $\sigma$ and an agent $\kappa$ that satisfy the two conditions in Definition \ref{def:extend}. Let $ \widehat{\ell} \coloneqq \sigma(\kappa)$ and note that condition $(ii)$ in Definition \ref{def:extend} implies that $\widehat{\ell} \in \M(p)$ along with $v_\kappa(A_{\widehat{\ell}} \cup \{g\}) - v_\kappa(A_{\widehat{\ell}})=1$. Hence, $k = \kappa$ and $\ell = \widehat{\ell}$ is a feasible choice for \textsc{Extend} in Line \ref{line:select-kl}. For this choice, the if-condition in Line \ref{line:return-ex} will hold: since with $k = \kappa$ and $\ell = \widehat{\ell} = \sigma(\kappa)$, the restriction of $\sigma$ from $[n] \setminus \{ k \}$ to $[n] \setminus \{ \ell \}$ is a matching in the bipartite graph $H$ (see Line \ref{line:def-bipartite}). Hence, the computed permutation $\rho$ satisfies $\sum_{i=1}^n  A_{\rho(i)} \geq \sum_{i=1}^n A_{\sigma(i)}$. Now, given that $\A_\sigma$ is an envy-freeable allocation (condition $(i)$ of Definition \ref{def:extend}), Theorem \ref{thm:hs_ef} gives us $\sum_{i=1}^n  v_i(A_{\sigma(i)}) \geq \sum_{i=1}^n v_i(A_i)$. Hence, the if-condition in Line \ref{line:return-ex} will execute successfully returning a (correct) permutation and agent for the input $(\A,p)$. This completes the proof.   
\end{proof}

\subsection{Proof of Lemma \ref{lem:s}}
\label{appendix:prop-non-ext}

This section restates and proves Lemma \ref{lem:s}. 

\LemmaFindSinkInduct*
\begin{proof}
For a non-extendable input $(\A,p)$ and $g$, write $s^0$ to denote the first agent considered by \textsc{FindSink} (Line \ref{line:find-sink-initialize}). Also, write $\mathcal{X}^0$ to denote the first allocation consider in the subroutine, $\mathcal{X}^0= (A_1, \ldots, A_{s^0} \cup \{g\}, \ldots, A_n)$. Furthermore, let $s^\tau$ and $\mathcal{X}^\tau$, respectively, denote the agent and allocation considered in the $\tau$th iteration of \textsc{FindSink}; in particular, $\mathcal{X}^\tau = (A_1, \ldots, A_{s^\tau} \cup \{g\}, \ldots, A_n)$. We will show, via induction, that, $s^\tau \in \M(p)$ and $\mathcal{X}^\tau$ is envy-freeable, for all $\tau \geq 0$. \\

\noindent
\emph{Base Case:} First, we note that, by construction, $s^0 \in \M(p)$ (Line \ref{line:find-sink-initialize}). Next, we show that $\mathcal{X}^0$ is envy-freeable. 
Assume towards a contradiction, that $\mathcal{X}^0$ is not envy-freeable. In such a case, Theorem \ref{thm:hs_ef} implies that there exists a permutation $\sigma$ such that the social welfare of allocation $\mathcal{X}^0_\sigma$ is strictly greater than that of $\mathcal{X}^0$. For the permutation $\sigma$, write allocation $\mathcal{B} = (B_1, \ldots, B_n) \coloneqq \A_\sigma$. Note that for agent $k \in [n]$, with the property that $\sigma(k) = s^0$, we have $B_k = A_{s^0}$; in particular, $\mathcal{X}^0_\sigma = (B_1, \ldots, B_k \cup \{g\}, \ldots, B_n)$. Also, note that the social welfare of $\mathcal{X}^0$ is at least the social welfare of $\mathcal{A}$; the agents' valuations are monotonic. These bounds between the social welfares of $\mathcal{X}^0_\sigma$, $\mathcal{X}^0$, and $\A$ imply\footnote{Recall that the valuations are dichotomous and, hence, integer valued.} 
\begin{align}
\sum_{i \neq k} v_i(B_i) + v_k(B_k \cup \{g \}) \geq \sum_{i=1}^n v_i(A_i) + 1  \label{ineq:SW-BA}
\end{align} 
In addition, since $\A$ is envy-freeable, Theorem \ref{thm:hs_ef} gives us  
\begin{align}
\sum_{i=1}^n v_i(A_i) \geq \sum_{i =1}^n v_i(B_i) \label{ineq:SW-AB}
\end{align}    
Here, if either inequality (\ref{ineq:SW-BA}) or (\ref{ineq:SW-AB}) is strict, then that would contradict the fact that valuation $v_k$ is dichotomous; specifically, that would contradict the bound $v_k(B_k \cup \{ g \}) - v_k(B_k) \leq 1$. Hence, both inequalities (\ref{ineq:SW-BA}) and (\ref{ineq:SW-AB}) hold with an equality. In particular, since equation (\ref{ineq:SW-AB}) is tight, we get that allocation $\mathcal{B}$ is envy-freeable;  analogous to $\A$ it maximizes social welfare among all reassignments.  Furthermore, combining equations (\ref{ineq:SW-BA}) and (\ref{ineq:SW-AB}), we obtain $v_k(B_k \cup \{ g\}) - v_k(B_k) = 1$. These observations, however, contradict the lemma assumption that $(\A,p)$ is not extendable: (i) $\mathcal{B} = \A_\sigma$ is  envy-freeable and, hence, $(\mathcal{B}, q)$ is an envy-free solution with $q \coloneqq p_\sigma$ (Lemma \ref{lem:reshuffle}), and (ii) agent $k$ satisfies $v_k(B_k \cup \{ g\}) - v_k(B_k) = 1$ and $k \in \M(q)$; here, we use the facts that $s^0 \in \M(p)$ and $q_k = p_{s^0}$ (since $\sigma(k) = s^0$). That is, permutation $\sigma$ and agent $k$ satisfy the extendability criteria (Definition \ref{def:extend}). Therefore, by way of contradiction, we get that $\mathcal{X}^0$ is an envy-freeable allocation. The envy-freeability of $\mathcal{X}^0$ along with the containment $s^0 \in \M(p)$ gives us the base case. \\

 \noindent
\emph{Induction Step:} Assuming that agent $s^{\tau-1} \in \M(p)$ and allocation $\mathcal{X}^{\tau-1}$ is envy-freeable, we establish the induction step for iteration count $\tau \geq 1$. Specifically, we will first prove that agent $s^\tau \in \M(p)$. 

The selection criterion in the while-loop of \textsc{FindSink} (Line \ref{loop_FS}) implies that the, under the subsidy vector computed for the envy-freeable allocation $\mathcal{X}^{\tau-1}$, agent $s^\tau$ must have required a subsidy of at least $2$. That is, in the envy graph $G_{\mathcal{X}^{\tau-1}}$ and starting at $s^\tau$, there exists a path with weight at least $2$ (Theorem~\ref{thm:hs_sub}). Now, assume towards a contradiction, that $s^\tau \notin \M(p)$. Since subsidy vector $p \in \{0,1\}^n$, it must be the case that $p_{s^\tau} =0$. Applying Theorem \ref{thm:hs_sub} again, we get that, in the envy-graph $G_\A$, the maximum-weight path starting at $s^\tau$ is of weight $0$. This, however, contradicts Proposition~\ref{lem:weight_increase}: the weight of any path in the envy-graph $G_{\mathcal{X}^{\tau-1}}$ can be at most one more than the weight of  the same path in $G_\mathcal{A}$. Hence, we must have $s^\tau \in \M(p)$. 

With the containment $s^\tau \in \M(p)$ in hand and using arguments analogous to the ones used in the base case (for $\mathcal{X}^0$), one can show that the allocation $\mathcal{X}^\tau$ is envy-freeable. This, overall, completes the induction step and establishes the lemma. 
\end{proof}

\input{efone-non}

%% file: efone-non.tex
\section{$\EFone$ Non-Example}
\label{appendix:efone}
Here, we provide a fair division instance for which a sample run of $\textsc{Alg}$ (Algorithm~\ref{Alg:BinSubsidy}) returns an allocation that is not $\EFone$. Consider an instance five goods, $\{g_1, g_2, g_3, g_4, g_5\}$, and three agents. For all subsets $S$ of goods, we define the agents' valuations as follows:
\begin{align*}
	v_1(S) & \coloneqq \min \big\{|S\cap \{g_1,g_4\}|, \ 1 \big\} \\
	v_2(S) & \coloneqq \left| S \cap \{g_1,g_3\} \right| + \min \big\{|S\cap\{g_2,g_4,g_5\}|, \ 1 \big\} \\
	v_3(S)& \coloneqq  |S \cap \{g_1\}| + \min \big\{|S\cap\{g_3,g_4,g_5\}|, \ 1 \big\} 
\end{align*}
Note that the valuations are dichotomous. 

Let the algorithm select the goods in the order of their indices, $g_1$ to $g_5$. The algorithm starts with the empty allocation $\mathcal{A}^1 = (\emptyset, \emptyset, \emptyset)$ and subsidy vector $p^1 = (0, 0, 0)$. We next detail the five iterations of the algorithm, wherein each of the five goods are respectively assigned. \\

\noindent
\textit{Iteration 1:} The good $g_1$ is selected first. Agent $1$ is in $\M(p^1)$, and $v_1(A^1_1\cup\{g_1\})-v_1(A^1_1)=1$. The identity permutation $\sigma$ and agent $1$ satisfy the extendability criteria and, hence, $(\A^1,p^1)$ is extendable with $g_1$. The algorithm (in Step~\ref{line:update-extend}) assigns $g_1$ to agent $1$ to obtain allocation $\A^2=(A^2_1, A^2_2, A^2_3)$. At this point, agents $2$ and $3$ require a subsidy of $1$ to make the solution envy-free. The following table lists the agents' valuations for the bundles in $\A^2$ and the corresponding subsidies $p^2=(p^2_1, p^2_2, p^2_3)$. \\

\newcolumntype{L}{>{$}l<{$}}
\newcolumntype{C}{>{$}c<{$}}
\newcolumntype{R}{>{$}r<{$}}
\newcommand{\nm}[1]{\textnormal{#1}}

\begin{table} [ht!]
	\centering
	\begin{tabular}{LCCCCR}
		\toprule
		\multicolumn{1}{l}{} &
		\multicolumn{3}{c}{Valuations}    &
		\multicolumn{1}{c}{Subsidies}    \\ 
		\cmidrule(lr){2-4}
		\cmidrule(lr){5-5}
		
		&
		\multicolumn{1}{c}{$A^2_1 =\{g_1\}$} &
		\multicolumn{1}{c}{$A^2_2=\emptyset$} &
		\multicolumn{1}{c}{$A^2_3=\emptyset$} &
		\multicolumn{1}{c}{$p^2_i$} \\
		\midrule
		
		\nm{Agent 1} & 1  & 0 & 0 & 0 \\
		\nm{Agent 2} & 1  & 0 & 0 & 1 \\
		\nm{Agent 3} & 1  & 0 & 0 & 1 \\  
		\bottomrule
	\end{tabular}
	\caption{After the first iteration.} 
	\end{table}

\noindent
\textit{Iteration 2:} The second good $g_2$ is selected. For agent $2$ we have $v_2(A^2_2\cup\{g_2\})-v_2(A^2_2) = 1$ and this agent belongs to $\M(p^2)$. The identity permutation $\sigma$ and agent $2$ satisfy the extendability criteria and, hence, $(\A^2,p^2)$ is extendable with $g_2$. The algorithm (in Step~\ref{line:update-extend}) assigns $g_2$ to agent $2$ to obtain allocation $\A^3=(A^3_1, A^3_2, A^3_3)$. Agent $3$ requires a subsidy of $1$ to make the allocation $\A^3$ envy-freeable. Table \ref{table:iteration-two} lists the valuations that each agent $i$ has for the bundles after the assignment of good $g_2$ and the corresponding subsidies. \\

\begin{table} [ht!]
	\centering
	\begin{tabular}{LCCCCR} 
		\toprule
		\multicolumn{1}{l}{} &
		\multicolumn{3}{c}{Valuations}    &
		\multicolumn{1}{c}{Subsidies}    \\ 
		\cmidrule(lr){2-4}
		\cmidrule(lr){5-5}
		
		&
		\multicolumn{1}{c}{$A^3_1=\{g_1\}$} &
		\multicolumn{1}{c}{$A^3_2=\{g_2\}$} &
		\multicolumn{1}{c}{$A^3_3$=$\emptyset$} &
		\multicolumn{1}{c}{$p^3_i$} \\
		\midrule
		
		\nm{Agent 1} & 1  & 0 & 0 & 0 \\
		\nm{Agent 2} & 1  & 1 & 0 & 0 \\
		\nm{Agent 3} & 1  & 0 & 0 & 1 \\  
		\bottomrule
	\end{tabular}
	\caption{After the second iteration.} 
	\label{table:iteration-two}
\end{table}

\noindent
\textit{Iteration 3: }The good $g_3$ is subsequently selected. For agent $3$ we have $v_3(A^3_3\cup\{g_3\})-v_2(A^3_3) = 1$, and belongs to $\M(p^3)$. The identity permutation $\sigma$ and agent $3$ satisfy the extendability criteria and, hence, $(\A^3,p^3)$ is extendable with $g_3$. The algorithm (in Step~\ref{line:update-extend}) assigns $g_3$ to agent $3$ to obtain allocation $\A^4$. Table \ref{table:iteration-three} lists the valuations that each agent $i$ has for the bundles after the assignment of good $g_3$, and the corresponding subsidies $p^4_i$. \\

\begin{table} [ht!]
	\centering
	\begin{tabular}{LCCCCR}
		\toprule
		\multicolumn{1}{l}{} &
		\multicolumn{3}{c}{Valuations}    &
		\multicolumn{1}{c}{Subsidies}    \\ 
		\cmidrule(lr){2-4}
		\cmidrule(lr){5-5}
		
		&
		\multicolumn{1}{c}{$A^4_1=\{g_1\}$} &
		\multicolumn{1}{c}{$A^4_2=\{g_2\}$} &
		\multicolumn{1}{c}{$A^4_3=\{g_3\}$} &
		\multicolumn{1}{c}{$p^4_i$} \\
		\midrule
		
		\nm{Agent 1} & 1  & 0 & 0 & 0 \\
		\nm{Agent 2} & 1  & 1 & 1 & 0 \\
		\nm{Agent 3} & 1  & 0 & 1 & 0 \\  
		\bottomrule
	\end{tabular}
	\caption{After the third iteration.}
	\label{table:iteration-three}
\end{table}

\noindent
\textit{Iteration 4: }The good $g_4$ is selected next. Note that for each agent the marginal value of $g_4$, with respect to her own bundle is zero, i.e., $v_i(A^4_i\cup\{g_4\})-v_i(A^4_i)=0$ for all $i\in \{1,2,3\}$.  Also, the social welfare of any allocation obtained by reassigning the bundles $A^4_1, A^4_2, A^4_3$ is strictly less than the social welfare of $\A^4$. Therefore, $(\A^4, p^4)$ is non-extendable with $g_4$. Since $p^4=(0,0,0)$, all the agents are in $\M(p^4)$. Furthermore, adding the good $g_4$ to agent $1$'s bundle does not increase any agent's subsidy beyond one. Hence, agent $1$ is a feasible candidate to be returned by the $\textsc{FindSink}$ subroutine. The algorithm (in Step~\ref{line:sink-assign}) assigns the good $g_4$ to agent $1$ to obtain the allocation $\A^5=(A^5_1, A^5_2, A^5_3)$. Now, agents $2$ and $3$ require a subsidy of one towards envy-freeness. Table \ref{table:iteration-four} lists the valuations that each agent $i$ has for the bundles after the assignment of good $g_4$ and the subsidies. 

\begin{table} [h!]
	\centering
	\begin{tabular}{LCCCCR}
		\toprule
		\multicolumn{1}{l}{} &
		\multicolumn{3}{c}{Valuations}    &
		\multicolumn{1}{c}{Subsidies}    \\ 
		\cmidrule(lr){2-4}
		\cmidrule(lr){5-5}
		
		&
		\multicolumn{1}{c}{$A^5_1=\{g_1,g_4\}$} &
		\multicolumn{1}{c}{$A^5_2=\{g_2\}$} &
		\multicolumn{1}{c}{$A^5_3=\{g_3\}$} &
		\multicolumn{1}{c}{$p^5_i$} \\
		\midrule
		
		\nm{Agent 1} & 1  & 0 & 0 & 0 \\
		\nm{Agent 2} & 2  & 1 & 1 & 1 \\
		\nm{Agent 3} & 2  & 0 & 1 & 1 \\  
		\bottomrule
	\end{tabular}
	\caption{After the assignment of good $g_4$.}
\label{table:iteration-four}
\end{table}

\noindent
\textit{Iteration 5:} Finally, the good $g_5$ is selected. Note that for each agent $i$ the marginal value of $g_5$ with respect to her current bundle is zero, $v_i(A^5_i \cup \{g_5 \}) - v_i(A^5_i) = 0$.  However, for agent $2$, the marginal of the good $g_5$ with respect to the bundle $A^5_3$ is one: $v_2(A^5_3\cup\{g_5\})-v_2(A^5_3)=1$. Consider the permutation $\sigma$ where $\sigma(1)=2, \sigma(2)=3$, and $\sigma(3)=1$. Note that the allocation $\mathcal{B} \coloneqq A^5_\sigma$ is envy-freeable with subsidies $q = p_\sigma$. Therefore, the extendability criteria is satisfied by $\sigma$ and agent $2$. Table \ref{table:allocB} lists the valuation that each agent $i$ has for the bundles in $\mathcal{B} = (B_1, B_2, B_3)$ and the corresponding subsidies $q_i$s. 

\begin{table} [ht!]
	\centering
	\begin{tabular}{LCCCCR}
		\toprule
		\multicolumn{1}{l}{} &
		\multicolumn{3}{c}{Valuations}    &
		\multicolumn{1}{c}{Subsidies}    \\ 
		\cmidrule(lr){2-4}
		\cmidrule(lr){5-5}
		
		&
		\multicolumn{1}{c}{$B_1=\{g_2\}$} &
		\multicolumn{1}{c}{$B_2=\{g_3\}$} &
		\multicolumn{1}{c}{$B_3=\{g_1,g_4\}$} &
		\multicolumn{1}{c}{$q_i$} \\
		\midrule
		
		\nm{Agent 1} & 0  & 0 & 1 & 1 \\
		\nm{Agent 2} & 1  & 1 & 2 & 1 \\
		\nm{Agent 3} & 0  & 1 & 2 & 0 \\  
		\bottomrule
	\end{tabular}
	\caption{Envy-freeable allocation $\mathcal{B} \coloneqq \A^5_\sigma$ and subsidy vector $q \coloneqq p_\sigma$.}
\label{table:allocB}
\end{table}

The algorithm (in Step~\ref{line:update-extend}) assigns the good $g_5$ to the bundle $B_2$ and obtains the final allocation $\A^6$. Table \ref{table:iteration-six} lists the agents' valuations for the bundles in $\A^6=(A^6_1, A^6_2, A^6_3)$ and the subsidies $p^6_i$s. \\

\begin{table} [h!]
	\centering
	\begin{tabular}{LCCCCR}
		\toprule
		\multicolumn{1}{l}{} &
		\multicolumn{3}{c}{Valuations}    &
		\multicolumn{1}{c}{Subsidies}    \\ 
		\cmidrule(lr){2-4}
		\cmidrule(lr){5-5}
		
		&
		\multicolumn{1}{c}{$A^6_1=\{g_2\}$} &
		\multicolumn{1}{c}{$A^6_2=\{g_3,g_5\}$} &
		\multicolumn{1}{c}{$A^6_3=\{g_1,g_4\}$} &
		\multicolumn{1}{c}{$p^6_i$} \\
		\midrule
		
		\nm{Agent 1} & 0  & 0 & 1 & 1 \\
		\nm{Agent 2} & 1  & 2 & 2 & 0 \\
		\nm{Agent 3} & 0  & 1 & 2 & 0 \\  
		\bottomrule
	\end{tabular}
\caption{Returned solution.}
\label{table:iteration-six}
\end{table}

Overall, \textsc{Alg} returns the envy-free solution $(\A^6,p^6)$. 

Notably, agent $1$ envies agent $3$, even after the removal of any good: $v_1(A^6_1) < v_1(A^6_3 \setminus\{g\})$ for all $g\in A^6_3$. Therefore, the returned allocation $\A^6$ is not $\EFone$. \\

\noindent
\emph{Remark:} The dichotomous valuations in the instance at hand are, in fact, binary submodular. Therefore, applying the algorithm of Goko et al.~\cite{goko2021fair}, one would find here an allocation that is both envy-freeable (with $0$/$1$ subsidies) and $\EFone$. This observation highlights that our algorithm executes differently from that of Goko et al.~\cite{goko2021fair}.